\pdfoutput=1
\def\year{2015}
\documentclass[letterpaper]{article}
\usepackage{aaai16}
\usepackage{times}
\usepackage{helvet}
\usepackage{courier}
\usepackage{amsmath}
\usepackage{amsthm}
\usepackage{amssymb}
\usepackage{amsfonts}
\usepackage{algorithm}
\usepackage{algpseudocode}
\usepackage{graphics}
\usepackage{mathrsfs}
\usepackage[title]{appendix}
\def\tr{\mathrm{tr}}
\newcommand{\norm}[1]{\left\lVert#1\right\rVert}
\newtheorem{mythm}{Theorem}
\newtheorem{mydef}{Definition}
\newtheorem{mylem}{Lemma}

\begin{document}
\title{Wishart Mechanism for Differentially Private Principal Components Analysis}
\author{Wuxuan Jiang, Cong Xie \and Zhihua Zhang \\ Department of Computer Science and Engineering \\ Shanghai Jiao Tong University, Shanghai, China \\ jiangwuxuan@gmail.com, \{xcgoner, zhihua\}@sjtu.edu.cn}

\maketitle
\begin{abstract}
\begin{quote}
	We propose a new input perturbation mechanism for publishing a covariance matrix to achieve $(\epsilon,0)$-differential privacy. Our mechanism uses a Wishart distribution to generate matrix noise. In particular, we apply this mechanism to principal component analysis (PCA). Our mechanism is able to keep the positive semi-definiteness of the published covariance matrix. Thus, our approach gives rise to a general publishing framework for input perturbation of a symmetric positive semidefinite matrix. Moreover, compared with the classic Laplace mechanism, our method has better utility guarantee. To the best of our knowledge, the Wishart mechanism is the best input perturbation approach for $(\epsilon,0)$-differentially private PCA. We also compare our work with previous exponential mechanism algorithms in the literature and provide near optimal bound while having more flexibility and less computational intractability.
\end{quote}
\end{abstract}

\section{Introduction}
\label{sec:introd}

Plenty of machine learning tasks deal with sensitive information such as financial and medical data. A common concern regarding data security arises on account of the rapid development of data mining techniques. Several data privacy definitions are proposed in the literature. Among them differential privacy (DP) has been widely used~\cite{dwork2006calibrating}. Differential privacy controls the fundamental quantity
of information that can be revealed with changing one individual. Beyond a concept in database security, differential privacy has been used by many researchers to develop
privacy-preserving learning algorithms \cite{chaudhuri2009privacy,chaudhuri2011differentially,bojarski2014differentially}. Indeed, this class of algorithms is applied to a large number of machine learning models including logistic regression~\cite{chaudhuri2009privacy}, support vector machine~\cite{chaudhuri2011differentially}, random decision tree~\cite{bojarski2014differentially}, etc.
Accordingly, these methods can protect the raw data even though the output and algorithm itself are published.

Differential privacy (DP) aims to hide the individual information while keeping basic statistics of the whole dataset. A simple idea to achieve
this purpose is to add some special noise to the original model. After that, the attacker, who has two outputs generated by slightly different inputs, cannot distinguish whether the output change comes from the artificial noise or input difference. However, the noise might influence the quality of regular performance of the model. Thus, we should carefully trade off between privacy and utility.

No matter what the procedure is, a query, a learning algorithm, a game strategy or something else, we are able to define differential privacy if this procedure takes a dataset as input and returns the corresponding output. In this paper, we study the problem of designing differential private principal component analysis (PCA). PCA reduces the data dimension while keeping the optimal variance. More specifically, it finds a projection matrix by computing a low rank approximation to the sample covariance matrix of the given data points.

Privacy-preserving PCA is a well-studied problem in the literature~\cite{dwork2014analyze,hardt2012beating,hardt2013beyond,hardt2014noisy,blum2005practical,chaudhuri2012near,kapralov2013differentially}. It outputs a noisy projection matrix for dimension reduction while preserving the privacy of any single data point. The extant privacy-preserving PCA algorithms have been devised based on two major features: the notion of differential privacy and the stage of randomization. Accordingly, the privacy-preserving PCA algorithms can be divided into distinct categories.

The notion of differential privacy has two types:  $(\epsilon,0)$-DP (also called pure DP) and $(\epsilon,\delta)$-DP
(also called approximate DP). $(\epsilon,\delta)$-DP is a weaker version of $(\epsilon,0)$-DP as the
former allows the privacy guarantee to be broken with tiny probability (more precisely,  $\delta$). In the seminal work on privacy-preserving PCA~\cite{dwork2014analyze,hardt2012beating,hardt2013beyond,hardt2014noisy,blum2005practical}, the authors used the notion of $(\epsilon,\delta)$-DP. In contrast, there is only a few work~\cite{chaudhuri2012near,kapralov2013differentially}, which is based on $(\epsilon,0)$-DP.

In terms of the stage of randomization, there are two mainstream classes of approaches. The first is randomly computing the eigenspace \cite{hardt2013beyond,hardt2014noisy,chaudhuri2012near,kapralov2013differentially}. The noise is added in the computing procedure. An alternative way is directly adding noise to the covariance matrix. Then one runs the non-private eigenspace computing algorithm to produce the output. This class of approaches is called input perturbation \cite{blum2005practical,dwork2014analyze}. The input perturbation algorithms publish a noisy sample covariance matrix before computing the eigenspace. Thus, any further operation on the noisy covariance matrix does not violate privacy guarantee. So far as the flexibility is concerned, the input perturbation has better performance because it is not limited only to computing eigenspace. Besides, the input perturbation approach is efficient because it merely takes extra efforts on generating the noise. In view of these advantages, our mechanism for privacy-preserving PCA is also based on input perturbation.

\subsection{Related Work}

\citeauthor{blum2005practical} (2005) proposed an early input perturbation framework (named SULQ), and the parameters of noise are refined by \citeauthor{dwork2006calibrating} (2006). \citeauthor{dwork2014analyze} (2014) proved the state-of-the-art utility bounds for $(\epsilon,\delta)$-DP. \citeauthor{hardt2012beating} (2012) provided a better bound under the coherence assumption. In \cite{hardt2013beyond,hardt2014noisy}, the authors used a noisy power method to produce the principal eigenvector iteratively with removing the previous generated ones. \citeauthor{hardt2014noisy} (2014) provided a special case for $(\epsilon,0)$-DP as well.

\citeauthor{chaudhuri2012near} (2012) proposed the first useful privacy-preserving PCA algorithm for $(\epsilon,0)$-DP based on an exponential mechanism~\cite{mcsherry2007mechanism}. \citeauthor{kapralov2013differentially} (2013) argued that the algorithm in \cite{chaudhuri2012near} lacks convergence time guarantee and used heuristic tests to check convergence of the chain, which may affect the privacy guarantee. They also devised a mixed algorithm for low rank matrix approximation. However, their algorithm is quite complicated to implement and takes $O(d^6/\epsilon)$ running time.
Here $d$ is the dimension of the data point.  

Our work is mainly inspired by \citeauthor{dwork2014analyze} (2014). Since they provided the algorithms\ for $(\epsilon,\delta)$-DP, we seek the similar approach for $(\epsilon,0)$-DP with a different noise matrix design. As input perturbation methods, \citeauthor{blum2005practical} (2005) and \citeauthor{dwork2014analyze} (2014) both used the Gaussian symmetric noise matrix for privately publishing a noisy covariance matrix. A reasonable worry is that the published matrix might be no longer positive semidefinite, a normal attribute for a covariance matrix.

\subsection{Contribution and Organization}

In this paper we propose a new mechanism for privacy-preserving PCA that we call \emph{Wishart mechanism}.
The key idea is to add a Wishart  noise matrix to the original sample covariance matrix.
A Wishart matrix is always positive semidefinite, which in turn makes the perturbed
covariance matrix positive semidefinite.
Additionally,  Wishart matrix can be regarded as the scatter matrix of some random Gaussian vectors~\cite{GuptaN:Book:2000}.
Consequently, our Wishart mechanism equivalently adds Gaussian noise to the original data points.

Setting appropriate parameters of Wishart distribution, we derive the $(\epsilon,0)$-privacy guarantee (Theorem~\ref{thm:wishart}). Compared to the present Laplace mechanism, our Wishart mechanism adds less noise (Section~\ref{sec:analysis}), which implies our mechanism always has better utility bound. We also provide a general framework for choosing Laplace or Wishart input perturbation for $(\epsilon,0)$-DP in Section~\ref{sec:analysis}.

Not only using the Laplace mechanism as a baseline, we also conduct theoretical analysis to compare our work with
other privacy-preserving PCA algorithms based on the $(\epsilon,0)$-DP.
With respect to the different criteria, we provide sample complexity bound (Theorem~\ref{thm:rank1}) for comparison with ~\citeauthor{chaudhuri2012near} (2012) and derive the low rank approximation closeness when comparing to ~\citeauthor{kapralov2013differentially} (2013). Other than the principal eigenvector guarantee in \cite{chaudhuri2012near}, we have the guarantee for rank-$k$ subspace closeness (Theorem~\ref{thm:6}). With using a stronger definition of adjacent matrices, we achieve a $k$-free utility bound (Theorem~\ref{thm:kappro}). Converting the lower bound construction in \cite{chaudhuri2012near,kapralov2013differentially} into our case, we can see the Wishart mechanism is near-optimal.

The remainder of the paper is organized as follows. Section~\ref{sec:prelim} gives the notation and definitions used in our paper. Section~\ref{sec:alg} lists the baseline and our designed algorithms. Section~\ref{sec:analysis} provides the thorough analysis on privacy and utility guarantee of our mechanism together with comparison to several highly-related work. Finally, we conclude the work in Section~\ref{sec:conlusion}. Note that we put some proofs and more explanation into the supplementary material.

\section{Preliminaries}
\label{sec:prelim}

We first give some notation that will be used in this paper. Let $I_m$ denote the $m{\times} m$ identity matrix.
Given an $m{\times} n$ real matrix $Z=[Z_{ij}]$,  let its full singular value decomposition (SVD) as
$Z=U\Sigma V^T$, where $U \in \mathbb{R}^{m\times m}$ and $V\in \mathbb{R}^{n\times n}$ are orthogonal (i.e., $ U^T U = I_m$ and
$V^T V = I_n$), and $\Sigma \in \mathbb{R}^{m\times n}$ is a diagonal matrix  with the $i$th diagonal entry $\sigma_i$ being the $i$th largest singular value of $Z$. Assume that the rank of $Z$ is $\rho $ ($\leq \min(m, n)$). This implies that $Z$ has $\rho$ nonzero singular values.
Let $U_{k}$ and $V_{k}$ be the first $k$ ($<\rho$) columns of $U$ and $V$, respectively,
and $\Sigma_{k}$ be the $k\times k$ top sub-block of $\Sigma$.
Then the $m\times n$ matrix $Z_k = U_{k} \Sigma_{k} V_{k}^T$ is the best rank-$k$ approximation to $Z$.

The Frobenius norm of $Z$ is defined as $\norm{Z}_F=\sqrt{\sum\limits_{i,j} Z^2_{ij}} = \sqrt{\sum\limits_{i=1}^{\rho} \sigma^2_{i}}$,
the spectral norm is defined as $\norm{Z}_2 = \max\limits_{x\neq 0}\frac{\norm{Z x}_2}{\norm{x}_2} = \sigma_1$, the nuclear norm is define as $\norm{Z}_*=\sum\limits_{i=1}^{\rho} \sigma_{i}$, and the $\ell_{1,1}$ norm is defined as $\norm{Z}_{1,1}=\sum\limits_{i,j}|Z_{ij}|$.

Given a set of $n$ raw data points $X =[x_1,\cdots, x_n]$ where $x_i\in \mathbb{R}^d$,
we  consider the problem of publishing a noisy empirical sample covariance matrix for doing PCA.
Following previous work on privacy-preserving PCA, we also assume $\norm{x_i}_2\leq1$.
The standard PCA computes the  sample covariance matrix of the raw data $A=\frac{1}{n}XX^T=\frac{1}{n}\sum_{i=1}^{n}x_i x^T_i$.
Since $A$ is a $d\times d$ symmetric positive semidefinite matrix, its SVD is equivalent to the spectral decomposition. That is,
$A=V\Sigma V^T$.
PCA  uses $V_k$ as  projection matrix to compute the low-dimensional representation of raw data: $Y \triangleq V_k^T X$.

In this work we use Laplace and Wishart distributions, which are defined as follows.
\begin{mydef}
	A random variable $z$ is said to have a Laplace distribution $z \sim Lap(\mu,b)$, if its probability density function is
	\begin{equation*}
	p(z) = \frac{1}{2b} \exp(-\frac{|z-\mu|}{b}).
	\end{equation*}
\end{mydef}

\begin{mydef}[\cite{GuptaN:Book:2000}]
	A $d\times d$ random symmetric positive definite matrix $W$ is said to have a Wishart distribution $W \sim W_d(m, C)$, if its probability density function is
	\begin{equation*}
	p(W)= \frac{|W|^{\frac{m-d-1}{2}}}{2^{\frac{md}{2}}|C|^{\frac{m}{2}}\Gamma_d(\frac{m}{2})}\exp(-\frac{1}{2}\tr(C^{-1} W)),
	\end{equation*}
	where $m>d-1$ and $C$ is a $d\times d$ positive definite matrix.
\end{mydef}

Now we introduce the formal definition of differential privacy.
\begin{mydef}
	A randomized mechanism M takes a dataset $D$ as input and outputs a structure $s \in R$, where $R$ is the range of $M$.
	For any two adjacent datasets  $D$ and $\hat{D}$ (with only one distinct entry), $M$ is said to be $(\epsilon,0)$-differential private if for all $S\subseteq R$ we have
	\begin{equation*}
	\Pr\{ M(D)\in S \} \leq e^\epsilon \Pr\{ M(\hat{D})\in S \},
	\end{equation*}
	where $\epsilon>0$ is a small parameter controlling the strength of privacy requirement.
\end{mydef}

This definition actually sets limitation on the similarity of output probability distributions for the given similar inputs.
Here the adjacent datasets can have several different interpretations. In the scenario of privacy-preserving PCA, our definition is as follows.
Two datasets $X$ and $\hat{X}$  are \emph{adjacent} provided $X=[x_1, \dots, x_{i-1}, x_i, x_{i+1}, \dots, x_n]$  and $\hat{X}=[x_1, \dots, x_{i-1}, \hat{x}_i, x_{i+1}, \dots, x_n]$ for $x_i \neq \hat{x}_i$.
It should be pointed out that our definition of adjacent datasets is slightly different from ~\cite{kapralov2013differentially},
which leads to significant difference on utility bounds. We will give more specifically discussions  in Section~\ref{sec:analysis}.

We also give the definition of $(\epsilon, \delta)$-differential privacy.
This notion requires less privacy protection so that it often brings better utility guarantee.

\begin{mydef}
	A randomized mechanism $M$ takes a dataset as input and outputs a structure $s \in R$, where $R$ is the range of $M$.
	For any two adjacent datasets $D$ and $\hat{D}$ (with  only one distinct entry), $M$ is said to be $(\epsilon, \delta)$-differential private if for all $S\subseteq R$ we have
	\begin{equation*}
	\Pr[M(D)\in S]\leq e^\epsilon \Pr[M(\hat{D})\in S]+\delta.
	\end{equation*}
\end{mydef}

Sensitivity analysis
is a general approach to achieving differential privacy. The following definitions show the two typical kinds of sensitivity.

\begin{mydef}
	The $\ell_1$ sensitivity is defined as
	\begin{equation*}
	s_1(M)=\max_{d(D,\hat{D})=1}\norm{M(D)-M(\hat{D})}_1.
	\end{equation*}
	%\end{mydef}
	%\begin{mydef}
	The $\ell_2$ sensitivity is defined as
	\begin{equation*}
	s_2(M)=\max_{d(D,\hat{D})=1}\norm{M(D)-M(\hat{D})}_2.
	\end{equation*}
\end{mydef}

The sensitivity describes the possible largest change as a result of individual data entry replacement. The $\ell_1$ sensitivity is used in Laplace Mechanism for $(\epsilon,0)$-differential privacy, while the $\ell_2$ sensitivity is used in Gaussian Mechanism for $(\epsilon,\delta)$-differential privacy. We list the two mechanisms for comparison.

\begin{mythm}[Laplace Mechanism] \label{thm:lap}
	Let $\lambda>s_1(M)/\epsilon$. Add Laplace noise $Lap(0,\lambda)$ to each dimension of $M(D)$. This mechanism provides $(\epsilon,0)$-differential privacy.
\end{mythm}

\begin{mythm}[Gaussian Mechanism] \label{thm:gauss}
	For $c^2>2\ln(1.25/\delta)$, let $\sigma>cs_2(M)/\epsilon$. Add Gaussian noise $N(0,\sigma^2)$ to each dimension of $M(D)$. This mechanism provides $(\epsilon,\delta)$-differential privacy.
\end{mythm}

The above mechanisms are all perturbation methods. Another widely used method is exponential mechanism~\cite{mcsherry2007mechanism} which is based on sampling techniques.

\section{Algorithms}
\label{sec:alg}

First we look at a general framework of privacy-preserving PCA.
According to the definition of differential privacy, a privacy-preserving PCA takes the raw data matrix $X$ as input
and then calculates the sample covariance matrix $A=\frac{1}{n}XX^T$.
Finally, it computes the top-$k$ subspace of $A$ as the output.

The traditional approach adds noise to the computing procedure. For example, \citeauthor{chaudhuri2012near} (2012) and \citeauthor{kapralov2013differentially} (2013) used a sampling based mechanism during computing eigenvectors to obtain approximate results. Our mechanism adds noise in the first stage, publishing $A$ in a differential private manner. Thus, our mechanism takes $X$ as input and outputs $A$. Afterwards we follow the standard PCA to compute the top-$k$ subspace. This can be seen as a differential private preprocessing procedure.

Our baseline is the Laplace mechanism (Algorithm~\ref{alg:lap} and Theorem~\ref{thm:lap}). To the best of our knowledge, Laplace mechanism is the only input perturbation method for $(\epsilon,0)$-DP PCA.
Since this private procedure ends before computing the subspace, this shows $M(D)=\frac{1}{n}DD^T$ in sensitivity definition.

\begin{algorithm}[htb]
	\caption{Laplace input perturbation}
	\label{alg:lap}
	\begin{algorithmic}[1]
		\Require
		Raw data matrix $X\in\mathbb{R}^{d\times n}$;
		Privacy parameter $\epsilon$;
		Number of data $n$;
		\State Draw $\frac{d^2+d}{2}$ i.i.d. samples from $Lap(0,\frac{2d}{n\epsilon})$, then form a symmetric matrix $L$. These samples are put in the upper triangle part. Each entry in lower triangle part is copied from the opposite position.
		\State Compute $A=\frac{1}{n}XX^T$;
		\State Add noise $\hat{A}=A+L$;
		\Ensure $\hat{A}$;
	\end{algorithmic}
\end{algorithm}

Note that to make $\hat{A}$  be symmetric,  we use a symmetric matrix-variate Laplace distribution in Algorithm~\ref{alg:lap}.
However, this mechanism cannot guarantee the  positive semi-definiteness of $\hat{A}$, a desirable attribute for a covariance matrix.
This motivates us to use a Wishart noise  alternatively, giving rise to the Wishart mechanism
in Algorithm~\ref{alg:wishart}.

\begin{algorithm}[htb]
	\caption{Wishart input perturbation}
	\label{alg:wishart}
	\begin{algorithmic}[1]
		\Require
		Raw data matrix $X\in\mathbb{R}^{d\times n}$;
		Privacy parameter $\epsilon$;
		Number of data $n$;
		\State Draw a sample $W$ from $W_d(d+1,C)$, where $C$ has $d$ same eigenvalues equal to $\frac{3}{2n\epsilon}$;
		\State Compute $A=\frac{1}{n}XX^T$;
		\State Add noise $\hat{A}=A+W$;
		\Ensure $\hat{A}$;
	\end{algorithmic}
\end{algorithm}

\section{Analysis}
\label{sec:analysis}

In this section, we are going to conduct theoretical analysis of Algorithms~\ref{alg:lap} and \ref{alg:wishart}
under the framework of differential private matrix publishing. The theoretical support has two parts: privacy and utility guarantee. The former is the essential requirement for privacy-preserving algorithms and the latter tells how well the algorithm works against a non-private version. Chiefly, we list the valuable theorems and analysis. All the technical proofs omitted can be found in the supplementary material.

\subsection{Privacy guarantee}\label{sec:pg}

We first show that both algorithms satisfy privacy guarantee.
Suppose there are two adjacent datasets $X=[x_1, \dots, v, \dots, x_n]\in \mathbb{R}^{d\times n}$ and $\hat{X}=[x_1, \dots, \hat{v}, \dots, x_n] \in \mathbb{R}^{d\times n}$ where $v \neq \hat{v}$ (i.e.,  only $v$ and $\hat{v}$ are distinct).
Without loss of generality, we further assume that  each data vector has the $\ell_2$ norm at most 1.

\begin{mythm} \label{thm:alg1}
	Algorithm~\ref{alg:lap} provides $(\epsilon,0)$-differential privacy.
\end{mythm}

This theorem can be quickly proved by some simple derivations so we put the proof in the supplementary material.

\begin{mythm} \label{thm:wishart}
	Algorithm~\ref{alg:wishart} provides $(\epsilon,0)$-differential privacy.
\end{mythm}
\begin{proof}
	Assume the outputs for the adjacent inputs $X$ and $\hat{X}$ are identical (denoted $A+W_0$). Here $A=\frac{1}{n}XX^T$ and $\hat{A}=\frac{1}{n}\hat{X}\hat{X}^T$. We define the difference matrix $\Delta \triangleq  A-\hat{A}=\frac{1}{n}(vv^T-\hat{v}\hat{v}^T)$.
	Actually the privacy guarantee is to bound the following term:
	\begin{equation*}
	\begin{aligned}
	&\frac{p(A+W=A+W_0)}{p(\hat{A}+W=A+W_0)}=\frac{p(W=W_0)}{p(W=A+W_0-\hat{A})} \\
	&=\frac{p(W=W_0)}{p(W=W_0+\Delta)}
	\end{aligned}
	\end{equation*}
	As $W\sim W_d(d+1,C)$, we have that
	\begin{equation*}
	\begin{aligned}
	&\frac{p(W=W_0)}{p(W=W_0+\Delta)}
	=\frac{\exp[-\frac{1}{2}\tr(C^{-1}W_0)]}{\exp[-\frac{1}{2}\tr(C^{-1}(W_0+\Delta))]}\\
	&=\exp [\frac{1}{2}\tr(C^{-1}(W_0+\Delta))-\tr(C^{-1}W_0)] \\
	&=\exp [\frac{1}{2}\tr(C^{-1}\Delta)].
	\end{aligned}
	\end{equation*}
	Then apply Von Neumann's trace inequality: For matrices $A,B\in\mathbb{R}^{d\times d}$, denote their $i$th-largest singular value as $\sigma_i(\cdot)$. Then $|\tr(AB)|\leq\sum_{i=1}^{d}\sigma_i(A)\sigma_i(B)$. So that
	\begin{equation} \label{eqn:01}
	\begin{aligned}
	&\exp [\frac{1}{2}\tr(C^{-1}\Delta)]\leq \exp[\frac{1}{2}\sum_{i=1}^{d}\sigma_i(C^{-1})\sigma_i(\Delta)] \\
	&\leq\exp[\frac{1}{2}\norm{C^{-1}}_2\norm{\Delta}_*].
	\end{aligned}
	\end{equation}
	Since $\Delta=A-\hat{A}=\frac{1}{n}(vv^T-\hat{v}\hat{v}^T)$ has rank at most 2, and by singular value inequality $\sigma_{i+j-1}(A+B)\leq\sigma_i(A)+\sigma_j(B)$, we can bound $\norm{\Delta}_*$:
	\begin{equation*}
	\begin{aligned}
	&n\norm{\Delta}_* \leq \sigma_1(vv^T)+\sigma_1(-\hat{v}\hat{v}^T)+\max\{\sigma_1(vv^T) \\
	&+\sigma_2(-\hat{v}\hat{v}^T),\sigma_2(vv^T)+\sigma_1(-\hat{v}\hat{v}^T)\} \\
	&=\sigma_1(vv^T)+\sigma_1(\hat{v}\hat{v}^T)+\max\{\sigma_1(vv^T),\sigma_1(\hat{v}\hat{v}^T)\} \\
	&\leq 3\max\sigma_1(vv^T)=3\max\norm{vv^T}_2 \\
	&=3\max\norm{v}^2_2\leq 3.
	\end{aligned}
	\end{equation*}
	In Algorithm~\ref{alg:wishart}, the scale matrix $C$ in Wishart distribution has $d$ same eigenvalues equal to $\frac{3}{2n\epsilon}$, which implies $\norm{C^{-1}}_2=\frac{2n\epsilon}{3}$. Substituting  these terms in Eq. \eqref{eqn:01} yields
	\begin{equation*}
	\begin{aligned}
	&\frac{p(A+W=A+W_0)}{p(\hat{A}+W=A+W_0)}\leq\exp[\frac{1}{2}\norm{C^{-1}}_2\norm{\Delta}_*] \\
	&\leq\exp[\frac{1}{2}\cdot\frac{2n\epsilon}{3}\cdot\frac{3}{n}]=e^\epsilon.
	\end{aligned}
	\end{equation*}
\end{proof}

\subsection{Utility guarantee}

Then we give bounds about how far the noisy results are from optimal. Since the Laplace and Wishart mechanisms are both input perturbation methods, their analyses are similar.

In order to ensure privacy guarantee, we add a noise matrix to the input data. Such noise may have effects on the property of the original matrix. For input perturbation methods, the \emph{magnitude} of the noise matrix directly determines how large the effects are.
For example, if the \emph{magnitude} of the noise matrix is even larger than data, the matrix after perturbation is surely covered by noise. Better utility bound means less noise added. We choose the spectral norm of the noise matrix to measure its \emph{magnitude}. Since we are investigating the privacy-preserving PCA problem, the usefulness of the  subspace of the top-$k$ singular vectors is mainly cared.

The noise matrix in the Laplace mechanism is constructed with $\frac{d^2+d}{2}$ i.i.d random  variables of $Lap(2d/n\epsilon)$. Using the tail bound for an ensemble matrix in~\cite{terence2012topics}, we have that the spectral norm of the noise matrix in Algorithm~\ref{alg:lap} satisfies $\norm{L}_2=O(2d\sqrt{d}/n\epsilon)$ with high probability.

Then we turn to analyze the Wishart mechanism. We use the tail bound of the Wishart distribution in~\cite{zhu2012short}:
\begin{mylem}[Tail Bound of Wishart Distribution] \label{lem:02}
	Let $W\sim W_d(m,C)$. Then for $\theta\geq 0$, with probability at most $d\exp(-\theta)$,
	\begin{equation*}
	\lambda_1(W)\geq(m+\sqrt{2m\theta(r+2)}+2\theta r)\lambda_1(C)\,
	\end{equation*}
	where $r=\tr(C)/\norm{C}_2$.
\end{mylem}

In our settings that $r=d$ and $m=d+1$, we thus have that with probability at most $d\exp(-\theta)$, 
\begin{equation*}
\lambda_1(W)\geq(d+1+\sqrt{2(d+1)(d+2)\theta}+2\theta d)\lambda_1(C).
\end{equation*}

Let $\theta=c\log d(c>1)$. Then $d\exp(-\theta)=d^{1-c}$. So we can say with high probability
\begin{equation*}
\lambda_1(W)=O([d+1+\sqrt{2(d+1)(d+2)\theta}+2\theta d]\lambda_1(C)).
\end{equation*}
For convenience, we write
\begin{equation*}
\lambda_1(W)=O(d\log d\lambda_1(C))=O(3d\log d/2n\epsilon).
\end{equation*}

We can see that the spectral norm of noise matrix generated by the Wishart mechanism is $O(d\log d/n\epsilon)$
while the Laplace mechanism requires $O(d\sqrt{d}/n\epsilon)$. This implies that the Wishart mechanism adds less noise to obtain privacy guarantee. We list the present four input perturbation approaches for comparison. Compared to the state-of-the-art results about $(\epsilon,\delta)$ case~\cite{dwork2014analyze}, our noise magnitude of $O(\frac{d\log d}{n\epsilon})$ is obviously worse than their $O(\frac{\sqrt{d}}{n\epsilon})$. It can be seen as the utility gap between $(\epsilon,\delta)$-DP and $(\epsilon,0)$-DP.

\begin{table}[htbp]
	\centering
	\caption{Spectral norm of noise matrix in input perturbation.}	
	\begin{tabular}{|c|c|c|}
		\hline
		\hline
		Approach & Noise magnitude & Privacy \\
		\hline
		Laplace & $O(d\sqrt{d}/n\epsilon)$ & $(\epsilon,0)$ \\
		\hline
		\cite{blum2005practical} & $O(d\sqrt{d\log d}/n\epsilon)$ & $(\epsilon,\delta)$ \\
		\hline
		Wishart & $O(d\log d/n\epsilon)$ & $(\epsilon,0)$ \\
		\hline
		\cite{dwork2014analyze} & $O(\sqrt{d}/n\epsilon)$ & $(\epsilon,\delta)$ \\
		\hline
	\end{tabular}
\end{table}

\subsubsection{General framework}

We are talking about the intrinsic difference between the Laplace and Wishart mechanisms.
The key element is the difference matrix $\Delta$ of two adjacent matrices.
Laplace mechanism adds a noise matrix according to the $\ell_1$ sensitivity, which equals to $\max\norm{\Delta}_{1, 1}$. Thus, the spectral norm of noise matrix is $O(\max\norm{\Delta}_{1,1}\sqrt{d}/n\epsilon)$. When it comes to the Wishart mechanism, the magnitude of noise is determined by $\norm{C}_2$. For purpose of satisfying privacy guarantee, we take $\norm{C}_2=\omega(\max\norm{\Delta}_*/n\epsilon)$.
Then the spectral norm of noise matrix is $O(\max\norm{\Delta}_*d\log d/n\epsilon)$. Consequently, we obtain the following theorem.

\begin{mythm} \label{thm:comp}
	$M$ is a $d\times d$ symmetric matrix generated by some input.
	For two arbitrary adjacent inputs, the generated matrices are $M$ and $\hat{M}$.
	Let $\Delta=M-\hat{M}$. Using the Wishart mechanism to publish M in differential private manner works better if
	\begin{equation*}
	\frac{\max\norm{\Delta}_{1,1}}{\max\norm{\Delta}_*}=\omega(\sqrt{d}\log d);
	\end{equation*}
	otherwise the Laplace mechanism works better.
\end{mythm}

\subsubsection{Top-$k$ subspace closeness}

We now conduct comparison between our mechanism and the algorithm in \cite{chaudhuri2012near}. \citeauthor{chaudhuri2012near} (2012) proposed an exponential-mechanism-based method, which outputs the top-$k$ subspace by drawing a sample from the matrix Bingham-von Mises-Fisher distribution. \citeauthor{wang2013differential} (2013) applied this algorithm to private spectral analysis on graph and showed that it outperforms the Laplace mechanism for output perturbation.
Because of the scoring function defined, it is hard to directly sample from the original Bingham-von Mises-Fisher distribution. Instead, \citeauthor{chaudhuri2012near} (2012) used Gibbs sampling techniques to reach an approximate solution. However, there is no guarantee for convergence. They check the convergence heuristically, which may affect the basic privacy guarantee.

First we provide our result on the top-$k$ subspace closeness:
\begin{mythm} \label{thm:6}
	Let $\hat{V_k}$ be the top-$k$  subspace of $A+W$ in Algorithm 2. Denote the non-noisy subspace as $V_k$ corresponding to $A$. Assume $\sigma_1\geq\sigma_2\geq\cdots\geq\sigma_d$ are singular values of $A$. If $\sigma_k-\sigma_{k+1}\geq2\norm{W}_2$, then with high probability
	\begin{equation*}
	\norm{V_k V_k^T-\hat{V}_k \hat{V}_k^T}_F\leq\frac{2\sqrt{k}\norm{W}_2}{\sigma_{k}-\sigma_{k+1}}.
	\end{equation*}
\end{mythm}

We apply the well-known Davis-Kahan $\sin\theta$ theorem~\cite{davis1963rotation} to obtain this result. This theorem characterizes the usefulness of our noisy top-$k$ subspace. Nevertheless, \citeauthor{chaudhuri2012near} (2012) only provided the utility guarantee on the principal eigenvector. So we can only compare the top-$1$ subspace closeness, correspondingly.

Before the comparison, we introduce the measure in ~\cite{chaudhuri2012near}.
\begin{mydef}
	A randomized algorithm $\mathcal{A}(\cdot)$ is an $(\rho,\eta)$-close approximation to the top eigenvector if for all data sets $\mathcal{D}$ of n points, output a vector $\hat{v_1}$ such that
	\begin{equation*}
	\mathbb{P}(\langle\hat{v_1},v_1\rangle\geq\rho)\geq 1-\eta.
	\end{equation*}
\end{mydef}

Under this measure, we derive the sample complexity of the Wishart mechanism.
\begin{mythm} \label{thm:rank1}
	If $n\geq \frac{3(d+1+\sqrt{2(d+1)(d+2)\log \frac{d}{\eta}}+2d\log \frac{d}{\eta} )}{2\epsilon(1-\rho^2)(\lambda_{1}-\lambda_{2})}$ and $\rho\geq\frac{\sqrt{2}}{2}$,
	then the Wishart mechanism is a $(\rho,\eta)$-close approximation to PCA.
\end{mythm}

Because a useful algorithm should output an eigenvector making $\rho$ close to 1, our condition of $\rho\geq\frac{\sqrt{2}}{2}$ is quite weak. Comparing to the sample complexity bound of the algorithm in ~\cite{chaudhuri2012near}:
\begin{mythm}
	If $n\geq \frac{d}{\epsilon(1-\rho)(\lambda_1-\lambda_2)}(\frac{\log\frac{1}{\eta}}{d}+\log \frac{4\lambda_1}{(1-\rho^2)(\lambda_1-\lambda_2)})$,
	then the algorithm in ~\cite{chaudhuri2012near} is a $(\rho,\eta)$-close approximation to PCA.
\end{mythm}
Our result has a factor up to $\log d$ with dropping the term $\log\frac{\lambda_1}{\lambda_1-\lambda_2}$. Actually, the relationship between $d$ and  $\frac{\lambda_1}{\lambda_1-\lambda_2}$ heavily depends on the data. Thus, as a special case of top-$k$ subspace closeness, our bound for the top-$1$ subspace is comparable to ~\citeauthor{chaudhuri2012near}'s (2012).

\subsubsection{Low rank approximation}

Here we discuss the comparison between the Wishart mechanism and privacy-preserving rank-$k$ approximation algorithm proposed in \cite{kapralov2013differentially,hardt2014noisy}. PCA can be seen as a special case of low rank approximation problems. \citeauthor{kapralov2013differentially} (2013) combined the exponential and Laplace mechanisms to design a low rank approximation algorithm for a symmetric matrix, providing strict guarantee on convergence. However, the implementation of the algorithm contains too many approximation techniques and it takes $O(d^6/\epsilon)$ time complexity while our algorithm takes $O(kd^2)$ running time. \citeauthor{hardt2014noisy} (2014) proposed an efficient meta algorithm, which can be applied to $(\epsilon,\delta)$-differentially private PCA. Additionally, they provided a $(\epsilon,0)$-differentially private version.

We need to point out that the definition of adjacent matrix in privacy-preserving low rank approximation is different from ours (our definition is the same as \cite{dwork2014analyze,chaudhuri2012near}). In the definition \cite{kapralov2013differentially,hardt2014noisy}, two matrices $A$ and $B$ are called adjacent if $\norm{A-B}_2\leq 1$, while we restrict the difference to a certain form $vv^T-\hat{v}\hat{v}^T$. In fact, we make a stronger assumption so that we are dealing with a case of less sensitivity. This difference impacts the lower bound provided in \cite{kapralov2013differentially}.

For the consistence of comparison, we remove the term $\frac{1}{n}$ in Algorithm \ref{alg:wishart}, which means we use the $XX^T$ for PCA instead of $\frac{1}{n}XX^T$.  This is also used by \citeauthor{dwork2014analyze} (2014).

Applying Lemma 1 in ~\cite{achlioptas2001fast}, we can immediately have the following theorem:
\begin{mythm} \label{thm:kappro}
	Suppose the original matrix is $A=XX^T$ and $\hat{A_k}$ is the rank-$k$ approximation of output by the Wishart mechanism.
	Denote the $k$-th largest eigenvalue of A as $\lambda_k$. Then
	\begin{equation*}
	\norm{A-\hat{A}_k}_2\leq \lambda_{k+1}+O(\frac{d\log d}{\epsilon}).
	\end{equation*}
\end{mythm}

\citeauthor{kapralov2013differentially} (2013) provided a bound of $O(\frac{k^3 d}{\epsilon})$ and \citeauthor{hardt2014noisy}  (2014) provided $O(\frac{k^{\frac{3}{2}}d\log^2 d }{\epsilon})$ for the same scenario. If $k^3$ is larger than $\log d$, our algorithm will work better. Moreover, our mechanism has better bounds than that of \citeauthor{hardt2014noisy} (2014) while both algorithms are computationally efficient. \citeauthor{kapralov2013differentially} (2013) established a lower bound of $O(\frac{kd}{\epsilon})$ according to their definition of adjacent matrix. If replaced with our definition, the lower bound will become $O(\frac{d}{\epsilon})$. The details will be given in the in supplementary material. So our mechanism is near-optimal.

\section{Concluding Remarks}
\label{sec:conlusion}

We have studied the problem of privately publishing a symmetric matrix and provided an approach for choosing Laplace or Wishart noise properly.
In the scenario of PCA, our Wishart mechanism adds less noise than the Laplace, which leads to better utility guarantee. Compared with the privacy-preserving PCA algorithm in ~\cite{chaudhuri2012near}, our mechanism has reliable rank-$k$ utility guarantee while the former \cite{chaudhuri2012near} only has rank-1.
For rank-1 approximation we have the comparable performance on sample complexity.
Compared with the low rank approximation algorithm in \cite{kapralov2013differentially}, the bound of our mechanism does not depend on $k$.
Moreover, our method is more tractable computationally.
Compared with the tractable algorithm in ~\cite{hardt2014noisy}, our utility bound is better.

Since input perturbation only publishes the matrix for PCA, any other procedure can take the noisy matrix as input. Thus, our approach has more flexibility. While other entry-wise input perturbation techniques make the covariance not be positive semidefinite, in our case the noisy covariance matrix still preserves this property.

\section*{Acknowledgments}
We thank Luo Luo for the meaningful technical discussion. We also thank Yujun Li, Tianfan Fu for support on the early stage of the work. This work is supported by the National Natural Science Foundation of China
(No. 61572017) and the Natural Science Foundation of Shanghai City (No. 15ZR1424200). 

{
	\bibliography{literature}{}
	\bibliographystyle{aaai}
}

\newpage
\appendix
\appendixpage
\appendixpageon
\section{Proof of privacy guarantee}

The basic settings are the same as section~\ref{sec:pg}.

\subsection{Proof of Theorem ~\ref{thm:alg1}}
In order to prove Theorem~\ref{thm:alg1}, we first give the following lemma.

\begin{mylem} \label{lem:01}
	For mechanism $M(D)=\frac{1}{n}DD^T$, the $\ell_1$ sensitivity $s_1$ satisfies
	\begin{equation*}
		\frac{d}{n}<s_1(M)<\frac{2d}{n}.
	\end{equation*}
\end{mylem}

\begin{proof}
	Suppose $v=(p_1,\cdots,p_d)^T$ and $\hat{v}=(q_1,\cdots,q_d)^T$. Then the $\ell_1$ sensitivity of $M(D)$ can be converted to the following optimization problem:
	\begin{equation*}
		\begin{aligned}
			&s_1(M)=\max \frac{1}{n}\sum_{1\leq i,j\leq d} |p_i p_j-q_i q_j|, \\
			&\text{subject to} \sum_{i=1}^{d} p_i^2\leq 1, \sum_{i=1}^{d} q_i^2\leq 1.
		\end{aligned}
	\end{equation*}
	Setting $p_i=\frac{1}{\sqrt{d}}$ and $q_i=0$ for $i=1, \ldots, d$, we can have a lower bound $s_1(M)\geq \frac{d}{n}$.
	Then applying the triangle inequality,  we have the upper bound:
	\begin{equation*}
		\begin{aligned}
			&\sum_{1\leq i,j\leq d} |p_i p_j-q_i q_j| < \sum_{1\leq i,j\leq d} |p_i p_j|+|q_i q_j| \\
			&=2\sum_{1\leq i,j\leq d} |p_i p_j|\leq \frac{2d}{n}.
		\end{aligned}
	\end{equation*}
\end{proof}

Now applying Lemma~\ref{lem:01} to Theorem~\ref{thm:lap} immediately obtains the privacy guarantee for the Laplace mechanism.

\section{Proof of utility guarantee}

\subsection{Proof of Theorem ~\ref{thm:6}}
\begin{proof}
	We use the following two lemmas.
	\begin{mylem}[Davis-Kahan $\sin\theta$ theorem~\cite{davis1963rotation}] \label{lem:03} Let the $k$-th eigenvector of $A$ and $\hat{A}$ be $v_k$ and $\hat{v_k}$. Denote $P_k=\sum_{i=1}^{k}\limits v_kv_k^T$ and $\hat{P_k}=\sum_{i=1}^{k}\limits\hat{v_k}\hat{v_k}^T$. If $\lambda_k(A)>\lambda_{k+1}(\hat{A})$, then
		\begin{equation*}
			\norm{P_k-\hat{P_k}}_2\leq\frac{\norm{A-\hat{A}}_2}{\lambda_k(A)-\lambda_{k+1}(\hat{A})}.
		\end{equation*}
	\end{mylem}
	\begin{mylem} [Weyl's inequality] \label{lem:weyl} If $M$, $H$ and $P$ are $d \times d$ Hermitian matrices such that $M=H+P$. Let the $k$-th eigenvalues of $M$, $H$ and $P$ be $\mu_k$, $\nu_k$ and $\rho_k$, respectively. For $i\in[n]$, we have
		\begin{equation*}
			\nu_i+\rho_d\leq\mu_i\leq\nu_i+\rho_1.
		\end{equation*}
	\end{mylem}
	In our case, $A$ and $\hat{A}$ are both symmetric positive semidefinite (because of the property of  Wishart distribution). So the eigenvalues equal to singular values. Then we use Lemma~\ref{lem:weyl} with $A=H$ and $W=P$. We obtain
	\begin{equation*}
		\sigma_i(A+W)\leq\sigma_i(A)+\sigma_1(W)=\sigma_i(A)+\norm{W}_2.
	\end{equation*}
	Applying Lemma\ref{lem:03} with $A=A$ and $\hat{A}=A+W$ leads to
	\begin{equation*}
		\begin{aligned}
			&\norm{P_k-\hat{P_k}}_2=\norm{V_k V_k^T-\hat{V}_k \hat{V}_k^T}_2 \\
			&\leq\frac{\norm{W}_2}{\lambda_k(A)-\lambda_{k+1}(A+W)} \\
			&\leq\frac{\norm{W}_2}{\lambda_k(A)-\lambda_{k+1}(A)-\norm{W}_2} \\
			&=\frac{\norm{W}_2}{\sigma_k-\sigma_{k+1}-\norm{W}_2}.
		\end{aligned}
	\end{equation*}
	Under the assumption $\sigma_k-\sigma_{k+1}\geq2\norm{W}_2$,  we finally have
	\begin{equation*}
		\norm{V_k V_k^T-\hat{V}_k \hat{V}_k^T}_2\leq\frac{2\norm{W}_2}{\sigma_{k}-\sigma_{k+1}}.
	\end{equation*}
	Using the property
	\begin{equation*}
	\norm{V_k V_k^T-\hat{V}_k \hat{V}_k^T}_F\leq\sqrt{k}\norm{V_k V_k^T-\hat{V}_k \hat{V}_k^T}_2
	\end{equation*}
	we finish the proof.
\end{proof}

\subsection{Proof of Theorem ~\ref{thm:rank1}}
We are going to find the condition on sample complexity to satisfy $(\rho,\eta)$-close approximation.

\begin{proof}
	Set $k=1$ in Theorem~\ref{thm:6}. Then
	\begin{equation*}
		\begin{aligned}
			\norm{V_1 V_1^T-\hat{V}_1 \hat{V}_1^T}_F & =\norm{v_1 v_1^T-\hat{v}_1 \hat{v}_1^T}_F \\
			&=\tr [(v_1 v_1^T-\hat{v}_1 \hat{v}_1^T)(v_1 v_1^T-\hat{v}_1 \hat{v}_1^T)^T] \\
			&=2-2(v_1^T \hat{v}_1)^2\leq\frac{2\norm{W}_2}{\lambda_{1}-\lambda_{2}}.
		\end{aligned}
	\end{equation*}
	The condition $\lambda_{1}-\lambda_{2}\geq 2\norm{W}_2$ requires the last term to have a upper bound of 1, which implies $|v_1^T \hat{v}_1|\geq \frac{\sqrt{2}}{2}$.
	Let $\eta=d\exp(-\theta)$, which is $\theta = \log \frac{d}{\eta}$,
	we have that
	with probability $1-\eta$,
	\begin{equation*}
		\begin{aligned}
			&(v_1^T \hat{v}_1)^2 \geq 1-\frac{\norm{W}_2}{\lambda_{1}-\lambda_{2}} \\
			&= 1-\frac{(d+1+\sqrt{2(d+1)(d+2)\theta}+2\theta d)\lambda_1(C)}{\lambda_{1}-\lambda_{2}} \\
			&=1-\frac{3(d+1+\sqrt{2(d+1)(d+2)\log \frac{d}{\eta}}+2d\log \frac{d}{\eta} )}{2n\epsilon(\lambda_{1}-\lambda_{2})}.
		\end{aligned}
	\end{equation*}
	
	Under the condition \\ $n\geq \frac{3(d+1+\sqrt{2(d+1)(d+2)\log \frac{d}{\eta}}+2d\log \frac{d}{\eta} )}{2\epsilon(1-\rho^2)(\lambda_{1}-\lambda_{2})}$
	
	\begin{equation*}
		\begin{aligned}
			(v_1^T \hat{v}_1)^2 & \geq 1-(1-\rho^2)=\rho^2
		\end{aligned}
	\end{equation*}	
	
	Which yields $\Pr(v_1^T \hat{v}_1\geq \rho) \geq 1-\eta$.
\end{proof}

\section{Lower bound for low rank approximation}
We mainly follow the construction of \citeauthor{kapralov2013differentially} (2013) and make a slight modification to fit into our definition of adjacent matrices.
\begin{mylem} \label{lem:pack}
	Define $C^k_\delta(Y)=\{S\in \textbf{G}_{k,d}:\norm{YY^T-SS^T}_2\leq\delta\}$.
	For each $\delta>0$ there exists family $\mathcal{F}=\{Y^1,\cdots,Y^N\}$ with $N=2^{\Omega(k(d-k)\log 1/\delta)}$, where $Y^i\in \textbf{G}_{k,d}$ such that $C^k_\delta(Y^i)\cap C^k_\delta(Y^j)=\emptyset$ for $i\neq j$.
\end{mylem}

\begin{mythm}
	Suppose the original matrix is $A=XX^T$ and $\hat{A_k}$ is the rank-$k$ approximation of output by the any $\epsilon$-differential private mechanism.
	Denote the $k$-th largest eigenvalue of A as $\lambda_k$. Then
	\begin{equation*}
		\norm{A-\hat{A_k}}_2\leq\lambda_{k+1}+\Omega(d/\epsilon).
	\end{equation*}
\end{mythm}
\begin{proof}
	Take a set $\mathcal{F}=\{Y^1,\cdots,Y^N\}$ in Lemma~\ref{lem:pack}.
	Construct a series of matrices $A^i=\gamma Y_iY_i^T$ where $i\in[N]$. Then
	\begin{equation*}
		E_{A^i} \Big[\norm{A^i-\hat{A^i_k}}_2 \Big] \leq\delta\gamma.
	\end{equation*}
	Let $\hat{A^i_k}=\hat{Y_i}\hat{\Sigma}\hat{Y_i}^T$.  Then letting $\widetilde A^i_k=\hat{Y_i}\hat{Y_i}^T$, we have
	\begin{equation*}
		E_{A^i}\Big[\norm{A^i-\widetilde A^i_k}_2 \Big] \leq 2\delta\gamma.
	\end{equation*}
	Using Markov's inequality leads to
	\begin{equation*}
		\Pr_{A^i} \Big[\norm{A^i-\widetilde A^i_k}_2\leq 4\delta\gamma \Big] >\frac{1}{2}.
	\end{equation*}
	Here is the main difference between our definition and \citeauthor{kapralov2013differentially} (2013). They consider the distance from $A^i$ to $A^j$ is at most $2\gamma$ since $\norm{A^i}_2\leq \gamma$. In our framework, $A^i$ is a dataset consisting of $\gamma$ data groups, each one is $Y^i$. Changing $A^i$ to $A^j$ means replacing $\gamma k$ data points with brand new ones. So we consider the distance is at most $2\gamma k$.
	
	The algorithm should put at least half of the probability mass into $C^k_{4\delta}(Y^i)$. Meanwhile, to satisfy the privacy guarantee
	\begin{equation*}
		\frac{\Pr\{M(A_i)\in C^k_{4\delta}(Y^i)\}}{\Pr\{M(A_j)\in C^k_{4\delta}(Y^i)\}}\leq e^{2\gamma k\epsilon}
	\end{equation*}
	So $\Pr\{M(A_j)\in C^k_{4\delta}(Y^i)\} \geq \frac{1}{2}e^{-2\gamma k\epsilon}$, we have
	\begin{equation*}
		\frac{1}{2}e^{-2\gamma k\epsilon}\cdot 2^{\Omega(k(d-k)\log 1/\delta)}\leq 1
	\end{equation*}
	Which implies $\gamma=\Omega(d\log(1/\delta)/\epsilon)$ and completes our proof.
\end{proof}

\end{document}